\documentclass{llncs}

\usepackage{amsmath, amssymb, amsfonts}
\usepackage{colortbl, graphicx, mathrsfs}
\usepackage[pagebackref]{hyperref}
\usepackage{enumerate}
\usepackage{pstricks, pst-node}

\newtheorem{Rule}{Rule}

\begin{document}

\title{Linear Kernelizations for Restricted 3-Hitting Set Problems}

\author{Xuan Cai}

\institute{Shanghai Jiao Tong University\\
\email{caixuanfire@sjtu.edu.cn}}

\maketitle

\begin{abstract}
The 3-\textsc{Hitting Set} problem is also called the \textsc{Vertex
Cover} problem on 3-uniform hypergraphs. In this paper, we address
kernelizations of the \textsc{Vertex Cover} problem on 3-uniform
hypergraphs. We show that this problem admits a linear kernel in
three classes of 3-uniform hypergraphs. We also obtain lower and
upper bounds on the kernel size for them by the parametric duality.
\end{abstract}

\section{Introduction}
Let $C$ be a collection of subsets of a finite set $S$. A
\emph{hitting set} is a subset of $S$ that has a nonempty
intersection with every element of $C$. The \textsc{Hitting Set}
problem is to decide whether there is a hitting set with cardinality
at most $k$ for a given $(S,C)$. If the cardinality of every element
of $C$ is upper-bounded by a fixed natural number $d$, then the
corresponding problem would be a $d$-\textsc{Hitting Set} problem.
In the parameterized complexity \cite{dow:fel,nie,flu:gro}, it is
known that the $d$-\textsc{Hitting Set} problem is
\emph{fixed-parameter tractable}, i.e., solvable in time
$O(f(k)\cdot n^c)$ for some constant $c$ independent of $k$.

The 3-\textsc{Hitting Set} problem, a special case of the
$d$-\textsc{Hitting Set} problem, is a focal point of researches in
the parameterized complexity. Formally, it is defined as follows:

\begin{center}
\fbox{
\parbox[c]{8cm}{
\textsc{3-Hitting Set Problem}
\begin{description}
\item[Instance:]A collection $C$ of subsets with cardinality 3 of a finite set
$S$ and a nonnegative integer $k$.

\item[Parameter:]$k$.

\item[Problem:]Is there a hitting set $S'$ with $|S'|\leq k$.

\end{description} } }
\end{center}

Actually, $(S,C)$ could be regarded as a hypergraph such that $S$
and $C$ correspond to the sets of vertices and hyperedges
respectively. In this sense, the hitting set is equivalent to the
vertex cover, and the 3-\textsc{Hitting Set} problem is the same as
the \textsc{Vertex Cover} problem\footnote{\scriptsize{For
simplicity, we study the 3-\textsc{Hitting Set} problem from the
viewpoint of hypergraph theory.}} on 3-\emph{uniform} hypergraphs in
which every hyperedge consists of 3 vertices.

A parameterized problem $Q$ is a set of instances with the form
$(I,k)$ where $I$ is the input instance and $k$ is a nonnegative
integer. A \emph{kernelization} is a polynomial-time preprocessing
procedure that transforms one instance $(I,k)$ of a problem into
another $(I',k')$ of it such that
\begin{itemize}
\item $k'\leq k$.

\item $|I'|\leq f(k')$, where $f$ is a computable function.

\item $(I,k)$ is a ``yes''-instance if and only if $(I',k')$ is a
``yes''-instance.
\end{itemize}

Hence, a kernelization can shrink the instance until that its size
depends only on the parameter, so that we could find an algorithm,
at least a brute force one, to efficiently solve the problem on the
kernel $f(k')$. It is apparent that a problem is fixed-parameter
tractable if and only if it has a kernelization.

The development of kernelizations provides a new approach for
practically solving some NP-hard problems including the
\textsc{Vertex Cover} problem on 3-uniform hypergraphs. Presently,
these kernelizations have found their way into numerous applications
in many fields, e.g. bioinformatics
\cite{ruc:son,che:deh:rau:ste:tai}, computer networks
\cite{kuh:ric:wat:wel:zol} as well as software testing
\cite{jon:har}.

\subsection{Related Work}
Buss \cite{bus:gol} has given a kernelization with kernel size
$O(k^2)$ for the \textsc{Vertex Cover} problem on graphs by putting
``high degree elements'' into the cover. Similar to Buss' reduction,
Niedermeier and Rossmanith \cite{nie:ros} have proposed a cubic-size
kernelization for the \textsc{Vertex Cover} problem on 3-uniform
hypergraphs.

Fellows \emph{et al}.
\cite{abu:col:fel:lan:sut:sym,deh:fel:ros:sha,fel} have introduced
the crown reduction and obtained a $3k$-kernelization for the
\textsc{Vertex Cover} problem on graphs. Recently, Abu-Khzam
\cite{abu} has reduced further the kernel of this problem on
3-uniform hypergraphs to quadratic size by employing the crown
reduction.

It is known that an optimal solution to the linear programming for
the \textsc{Vertex Cover} problem on graphs can be transformed into
the \emph{half-integral} form, i.e., the variables take one of three
possible values $\{0,\frac{1}{2},1\}$. Based on this form, Nemhauser
and Trotter's Theorem
\cite{nem:tro,bar:eve,abu:col:fel:lan:sut:sym,khu} guarantees the
\textsc{Vertex Cover} problem has a $2k$-kernelization on graphs. We
also wish to follow the above idea to get a linear kernelization of
the \textsc{Vertex Cover} problem on 3-uniform hypergraphs. However,
it seems impossible to achieve the goal because
Lov$\mathrm{\acute{a}}$sz \cite{lov} and Chung \emph{et al}.
\cite{chu:fur:gar:gra} have shown that the optimal solution to the
linear programming has no such \emph{half-integral} form.

Very recently, Fellows \emph{et al}. \cite{bod:dow:fel:her} propose
a new method which allows us to show that many problems do not have
polynomial size kernels under reasonable complexity-theoretic
assumptions.

\subsection{Our Work}
In this paper, we show that the \textsc{Vertex Cover} problem in
three classes of 3-uniform hypergraphs has a linear kernelization.
Moreover, we provide lower and upper bounds on the kernel size for
them by the parametric duality.

The rest of this paper is organized as follows. In Section 2, we
introduce some definitions and notations. In Section 3, we study the
\textsc{Vertex Cover} problem on quasi-regularizable 3-uniform
hypergraphs, and show that this problem has a linear kernelization.
In Section 4 and 5, we also show this problem admits a linear
kernelization on bounded-degree and planar 3-uniform hypergraphs
respectively. Furthermore, we present lower and upper bounds on the
kernel size for them. Finally, we present some questions left open
and discuss the future work in Section 6.

\section{Preliminaries}
Let $\mathcal{H}=(\mathcal{V},\mathcal{E})$ be a hypergraph with
$\mathcal{V}=\{v_1,\cdots,v_n\}$ and
$\mathcal{E}=\{B_1,\cdots,B_m\}$. For a hyperedge $B_i$, if
$|B_i|=1$, then we call it a \emph{self-loop}. If there exists
another hyperedge $B_j$ with $B_j\subset B_i$, then we call $B_i$
\emph{dominated} by $B_j$. For a vertex $v$,
$E(v)=\{B_i\in\mathcal{H}\ |\ v\in B_i\}$ and $|E(v)|$ are the
\emph{incidence set} and the \emph{degree} of $v$ respectively.
Similarly, $E(A)=\bigcup_{v\in A}E(v)$ for a vertex subset $A$. If
there exists another vertex $w$ with $E(w)\subseteq E(v)$, then we
call $w$ \emph{dominated} by $v$. If $w,v\in B_i$, then we call $v$
an \emph{adjacent} vertex of $w$.

Similar to the definition of regular graphs, if all the vertices of
a hypergraph $\mathcal{H}$ have the same degree, then we call
$\mathcal{H}$ \emph{regular}.

\begin{definition}
If the resulting hypergraph is regular after replacing each
hyperedge $B_i$ of $\mathcal{H}$ with $k_i$-multiple hyperedges
($k_i\geq0$), then $\mathcal{H}$ is quasi-regularisable.
\end{definition}

Here, we define $k$-\emph{multiple} edges $B_i$ $k$ copies of $B_i$.
Note that $k=0$ means that $B_i$ is removed from the hypergraph.

A planar graph can be embedded on a plane without edge intersections
except at the endpoints. Analogously, we can define a planar
hypergraph.

\begin{definition}
$G_{\mathcal{H}}=(V_1\cup V_2,E)$ is a bipartite incidence graph of
the hypergraph $\mathcal{H}=(\mathcal{V},\mathcal{E})$ if
$G_{\mathcal{H}}$ satisfies the following conditions:
\begin{enumerate}
\item $V_1=\mathcal{V}$.

\item $V_2=\{v_B\ |\ B\in\mathcal{E}\}$.

\item $E=\{\{v,v_B\}\ |\ v\in V_1\text{ and }v_B\in V_2\text{ and }v\in
B\}$.
\end{enumerate}
\end{definition}

\begin{definition}
A hypergraph $\mathcal{H}$ is planar if $G_{\mathcal{H}}$ is planar.
\end{definition}

The optimization version of the \textsc{Vertex Cover} problem is
often concerned in kernelizations
\cite{abu:col:fel:lan:sut:sym,flu:gro}. Assigning a 0-1 variable for
each vertex, it is easy to establish the integer programming for the
optimal version of the \textsc{Vertex Cover} problem on
$\mathcal{H}$ as follows:

\begin{equation}\label{hvc}
\begin{array}{rll}
\min&\sum_{i=1}^nx_i\\
\text{subject to:}&x_i+x_j+x_k\geq1,&\{v_i,v_j,v_k\}\in\mathcal{E};\\
&x_i\in\{0,1\},&v_i\in\mathcal{V}.
\end{array}
\end{equation}
The optimal value of (\ref{hvc}) is called the \emph{node-covering
number} of $\mathcal{H}$, denoted by $\tau(\mathcal{H})$. Relaxing
the restriction to rational number field, we can obtain the
corresponding linear programming:
\begin{equation}\label{hlp}
\begin{array}{rll}
\min&\sum_{i=1}^nx_i\\
\text{subject to:}&x_i+x_j+x_k\geq1,&\{v_i,v_j,v_k\}\in\mathcal{E};\\
&0\leq x_i\leq1&v_i\in\mathcal{V}.
\end{array}
\end{equation}
The optimal value of (\ref{hlp}) is called the \emph{fractional
node-covering number} of $\mathcal{H}$, denoted by
$\tau^*(\mathcal{H})$.

\section{Quasi-regularizable 3-uniform Hypergraphs}
In this section, we consider quasi-regularizable 3-uniform
hypergraphs and show that the \textsc{Vertex Cover} problem has a
linear kernelization on them.

\begin{theorem}
Let $\mathcal{H}$ be a quasi-regularizable 3-uniform hypergraph. The
\textsc{Vertex Cover} problem admits a problem kernel of size $3k$
on $\mathcal{H}$.
\end{theorem}

\begin{proof}
Let $\mathcal{H}=(\mathcal{V},\mathcal{E})$ with
$\mathcal{V}=\{v_1,\cdots,v_n\}$ and
$\mathcal{E}=\{B_1,\cdots,B_m\}$. Since $\mathcal{H}$ is
quasi-regularizable, we can transform it into an $r$-regular
hypergraph $\mathcal{H}'=(\mathcal{V},\mathcal{E}')$ with $r>0$. If
$m'$ denotes the number of hyperedges in $\mathcal{H}'$, then
$m'=\frac{1}{3}rn$.

Relaxing (\ref{hvc}), we establish the linear programming for the
\textsc{Vertex Cover} problem on $\mathcal{H}'$,
\begin{equation}\label{qr}
\begin{array}{rll}
\min&\sum_{i=1}^nx_i\\
\text{subject to:}&x_i+x_j+x_k\geq1,&\{v_i,v_j,v_k\}\in\mathcal{E}';\\
&0\leq x_i\leq1&v_i\in\mathcal{V}.
\end{array}
\end{equation}
It is easy to see that $x_1=\cdots=x_n=\frac{1}{3}$ is a feasible
solution. Thus, we have
\begin{equation}\label{ie1}
\tau^*(\mathcal{H}')\leq\sum_{i=1}^nx_i=\frac{1}{3}n.
\end{equation}

On the other hand, let $(x_1^*,\cdots,x_n^*)$ be an optimal solution
of (\ref{qr}). Note that the optimal solution must satisfy all the
restrictions. Summing up all the inequalities in (\ref{qr}), it is
not difficult to get
\begin{equation}\label{ie2}
r\cdot\tau^*(\mathcal{H}')=r\sum_{i=1}^nx_i^*\geq m'=\frac{1}{3}nr
\end{equation}

By (\ref{ie1}) and (\ref{ie2}), we have
$\tau^*(\mathcal{H}')=\frac{1}{3}n$. If $k<\tau(\mathcal{H})$, then
$\mathcal{H}$ has no vertex cover with size at most $k$. Therefore,
we have $n\leq3k$ because
$\tau^*(\mathcal{H})=\tau^*(\mathcal{H}')\leq\tau(\mathcal{H})=\tau(\mathcal{H}')\leq
k$.\qed
\end{proof}

\iffalse
 Actually, we have the following stronger theorem.

\begin{theorem}[\cite{ber}]
Let $\mathcal{H}=(\mathcal{V},\mathcal{E})$ be an $r$-uniform
hypergraph with $|\mathcal{V}|=n$. $\mathcal{H}$ is
quasi-regularisable if and only if
$\tau^*(\mathcal{H})=\frac{n}{r}$.
\end{theorem}

\begin{corollary}
For any $\varepsilon>0$, there is no
$(\frac{3}{2}-\varepsilon)k$-kernel for the \textsc{Independent Set}
problem on quasi-regularizable 3-uniform hypergraphs unless P$=$NP.
\end{corollary}
\fi

\section{Bounded-degree 3-uniform Hypergraphs}
In this section, we give linear kernelizations for the
\textsc{Vertex Cover} problem and its dual problem on bounded-degree
3-uniform hypergraphs, respectively. Meanwhile, we derive lower
bounds on kernel sizes for these two problem.

\begin{theorem}\label{tb}
Let $\mathcal{H}$ be a 3-uniform hypergraph with bounded degree $d$.
The \textsc{Vertex Cover} problem admits a problem kernel of size
$dk$ on $\mathcal{H}$.
\end{theorem}

\begin{proof}
Let $\mathcal{H}=(\mathcal{V},\mathcal{E})$ with
$\mathcal{V}=\{v_1,\cdots,v_n\}$ and
$\mathcal{E}=\{B_1,\cdots,B_m\}$. Assume that $d_i$ is the degree of
$v_i$. Then, we establish the integer programming (\ref{hvc}) and
the linear programming (\ref{hlp}) for the \textsc{Vertex Cover}
problem on $\mathcal{H}$.

If $k<\tau(\mathcal{H})$, then $\mathcal{H}$ has no vertex cover
with size at most $k$. Summing up all the inequalities in
(\ref{hlp}),
\[
m\leq\sum_{i=1}^m \sum_{v_j\in B_i}x_j=\sum_{i=1}^nd_ix_i\leq
d\sum_{i=1}^nx_i
\]
Since $\tau^*(\mathcal{H})\leq\tau(\mathcal{H})\leq k$, we have
$m\leq d\cdot\tau^*(\mathcal{H})\leq d\cdot\tau(\mathcal{H})\leq
d\cdot k$. \qed
\end{proof}

Chen \emph{et al}. \cite{che:fer:kan:xia} have studied the
parametric duality and kernelization. They have proposed the
lower-bound technique on the kernel size.

\begin{theorem}[\cite{che:fer:kan:xia}]\label{dual}
Let $(P,s)$ be an NP-hard parameterized problem. Suppose that $P$
admits an $\alpha k$ kernelization and its dual $P_d$ admits an
$\alpha_dk_d$ kernelization, where $\alpha,\alpha_d\geq1$. If
$(\alpha-1)(\alpha_d-1)<1$, then P$=$NP.
\end{theorem}
Here, $k_d=s(I)-k$ and
$s:\sigma^*\times\mathbb{N}\rightarrow\mathbb{N}$ is a \emph{size
function} for a parameterized problem $P$ if
\begin{itemize}
\item $0\leq k\leq s(I,k)$

\item $s(I,k)\leq|I|$

\item $s(I,k)=s(I,k')$ for all appropriate $k,k'$.
\end{itemize}

For the \textsc{Vertex Cover} problem, it is clear that its dual
problem is \textsc{Independent Set} problem.

\begin{center}
\fbox{
\parbox[c]{8cm}{
\textsc{Independent Set Problem}
\begin{description}
\item[Instance:]A hypergraph $\mathcal{H}=(\mathcal{V},\mathcal{E})$
and a nonnegative integer $k$.

\item[Parameter:]$k$.

\item[Problem:]Is there a vertex subset $I\subseteq \mathcal{V}$
such that $|I|\geq k$ and each hyperedge contains no two vertices of
$I$.
\end{description} } }
\end{center}

\begin{theorem}\label{tbi}
Let $\mathcal{H}$ be a 3-uniform hypergraph with bounded degree $d$.
The \textsc{Independent Set} problem admits a problem kernel of size
$(2d+1)dk$ on $\mathcal{H}$.
\end{theorem}

\begin{proof}
Without loss of generality, let $I$ be a maximum independent set of
$\mathcal{H}$ and $|I|\leq k$. Since $\mathcal{H}$ is a 3-uniform
hypergraph with bounded degree $d$, the number of its adjacent
vertices is bounded by $2d$ for each vertex $v\in I$, and the number
of hyperedges among these adjacent vertices is bounded by $2d^2$.
Therefore, we have
\[
|\mathcal{E}|\leq k\cdot d+2d^2\cdot k=(2d+1)dk
\]
\qed
\end{proof}

By Theorem~\ref{dual}, we easily get lower bounds on kernel sizes
for the \textsc{Vertex Cover} problem and the \textsc{Independent
Set} problem on 3-uniform hypergraphs with bounded degree $d$.

\begin{corollary}
For any $\varepsilon>0$, there is no
$(\frac{2d^2+d}{2d^2+d-1}-\varepsilon)k$-kernel for the
\textsc{Vertex Cover} problem on 3-uniform hypergraphs with bounded
degree unless P$=$NP.
\end{corollary}

\begin{corollary}
For any $\varepsilon>0$, there is no
$(\frac{d}{d-1}-\varepsilon)k$-kernel for the \textsc{Independent
Set} problem on 3-uniform hypergraphs with bounded degree unless
P$=$NP.
\end{corollary}

\section{Planar 3-uniform Hypergraphs}
In this section, we exhibit three equivalent parameterized problems
on planar 3-uniform hypergraphs. Building on them, we provide a
linear kernelization and a lower bound on the kernel size for the
\textsc{Vertex Cover} problem on planar 3-uniform hypergraphs.

\begin{theorem}\label{t1}
Let $\mathcal{H}$ be a planar 3-uniform hypergraph. The
\textsc{Vertex Cover} problem admits a problem kernel of size $67k$
on $\mathcal{H}$.
\end{theorem}

This proof consists of two parts. First, we construct two
parameterized problems and show that they are equivalent to the
\textsc{Vertex Cover} problem on planar 3-uniform hypergraphs. In
the second part, we show Theorem~\ref{t1} based on the known linear
kernelization \cite{alb:fel:nie:1,alb:fel:nie:2,che:fer:kan:xia} for
the \textsc{Dominating Set} on planar graphs.

\subsection{Equivalent Problems}
A \emph{dominating set} of a graph is a vertex subset $D$ such that
either $v\in D$ or $v$ is adjacent to some vertex of $D$ for each
vertex $v$ of the graph. $G_{\mathcal{H}}$ and $G^K_{\mathcal{H}}$
denote the bipartite incidence graph and the local complete graph of
a hypergraph $\mathcal{H}$.

\begin{definition}\label{def1}
$G^K_{\mathcal{H}}$ is the local complete graph of $\mathcal{H}$ if
$G^K_{\mathcal{H}}$ satisfies the following conditions:
\begin{enumerate}
\item $V(G^K_{\mathcal{H}})=V(G_{\mathcal{H}})=V_1\cup V_2$.

\item $E(G^K_{\mathcal{H}})=E(G_{\mathcal{H}})\cup\{\{x,y\}\ |\ \{x,y\}\subseteq B\in\mathcal{E}\}$.
\end{enumerate}
\end{definition}

\begin{lemma}
If $\mathcal{H}$ is 3-uniform and planar, then $G^K_{\mathcal{H}}$
is planar.
\end{lemma}
\begin{proof}
It is trivial.\qed
\end{proof}

Here, we consider two parameterized problems. One parameterized
problem is the \textsc{Quasi-dominating Set} problem.

\begin{center}
\fbox{\parbox[c]{8cm}{\textsc{Quasi-dominating Set Problem}
\begin{description}
\item[Instance:]A bipartite graph $G=(V_1\cup
V_2,E)$ and a nonnegative integer $k$.

\item[Parameter:]$k$.

\item[Problem:]Is there $D\subseteq V_1$ such that
$\{v\ |\ \{v,w\}\in E\text{ and }w\in D\}=V_2$ and $|D|\leq k$?
\end{description}
}}
\end{center}

It is apparent that the \textsc{Vertex Cover} problem on
$\mathcal{H}$ is equivalent to the \textsc{Quasi-dominating Set}
problem on $G_{\mathcal{H}}$.

The other is the \textsc{Dominating Set} problem.

\begin{center}
\fbox{\parbox[c]{8cm}{\textsc{Dominating Set Problem}
\begin{description}
\item[Instance:]A graph $G=(V,E)$ and a nonnegative integer $k$.

\item[Parameter:]$k$.

\item[Problem:]Is there a dominating set $D\subseteq V$ with $|D|\leq
k$?
\end{description}
}}
\end{center}

\begin{lemma}
Let $\mathcal{H}$ be a 3-uniform hypergraph. $(G_{\mathcal{H}}^K,k)$
is a YES-instance of the \textsc{Dominating Set} problem if and only
if $(G_{\mathcal{H}},k)$ is a YES-instance of the
\textsc{Quasi-dominating Set} problem.
\end{lemma}

\begin{proof}
Let $D$ be a dominating set of $G_\mathcal{H}^K$. If $D\cap
V_2=\emptyset$, then $D$ is also a solution of the
\textsc{Quasi-dominating Set} problem on $G_{\mathcal{H}}$.

Otherwise, there is a vertex $v_B\in D\cap V_2$. By definition of
$G_{\mathcal{H}}^K$, $v_B$ corresponds to a hyperedge in
$\mathcal{H}$, says $B=\{u_1,u_2,u_3\}$. Figure~\ref{fig1}
illustrates that $D\backslash v_B\cup\{u_1\}$ is also a dominating
set of $G_{\mathcal{H}}^K$.

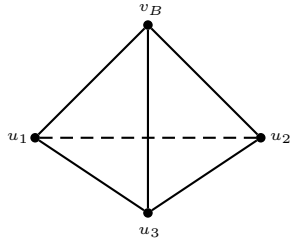
\begin{figure}
\begin{pspicture}(-3.5,-1.5)(13,1.6)

\dotnode(2.5,1.5){D}

\dotnode(1,0){A} \dotnode(4,0){C}

\dotnode(2.5,-1){E}

\uput[0](2.2,1.7){$\scriptstyle v_B$}

\uput[0](0.45,0){$\scriptstyle u_1$}

\uput[0](3.95,0){$\scriptstyle u_2$}

\uput[0](2.2,-1.25){$\scriptstyle u_3$}

\ncline{A}{D}\ncline{C}{D}

\ncline{A}{E}\ncline{C}{E} \ncline{D}{E}

\psset{linestyle=dashed} \ncline{A}{C}
\end{pspicture}
\caption{$v_B,u_1,u_2,u_3$ in $G_{\mathcal{H}}^K$}

\label{fig1}
\end{figure}

Thus, we can find a dominating set $D'\cap V_2=\emptyset$ of
$G_{\mathcal{H}}^K$ remaining the cardinality of $D$. $D'$ is a
solution of the \textsc{Quasi-dominating Set} problem on
$G_{\mathcal{H}}$.

The other direction of this proof is trivial. We complete the
proof.\qed
\end{proof}

\begin{corollary}
Let $\mathcal{H}$ be a 3-uniform hypergraph. The \textsc{Vertex
Cover} problem on $\mathcal{H}$, the \textsc{Dominating Set} problem
on $G_\mathcal{H}^K$ and the \textsc{Quasi-dominating Set} problem
on $G_{\mathcal{H}}$ are equivalent each other.
\end{corollary}

\subsection{Kernelization}
Alber \emph{et al}. \cite{alb:fel:nie:1,alb:fel:nie:2} first show
that the \textsc{Dominating Set} problem on planar graphs has a
problem kernel with size $335k$ based on two reduction rules. Chen
\emph{et al}. \cite{che:fer:kan:xia} extend those two reduction
rules and further reduce this upper bound on the kernel size to
$67k$.

\begin{theorem}[\cite{che:fer:kan:xia}]\label{tchen}
Let $G$ be a planar graph. The \textsc{Dominating Set} problem
admits a problem kernel of size $67k$ on $G$.
\end{theorem}

All the reduction rules is founded on neighborhoods of vertices. The
\emph{neighborhood} of $v$ in $G=(V,E)$ is $N(v)=\{x\in V\ |\
\{x,v\}\in E\}$. Then we can partition $N(v)$ into 3 disjoint sets
as follows,
\begin{itemize}
\item $N_1(v)=\{u\in N(v)\ |\ N(u)\backslash N[v]\neq\emptyset\}.$

\item $N_2(v)=\{u\in N(v)\backslash N_1(v)\ |\ N(u)\cap N_1(v)\neq\emptyset\}$.

\item $N_3(v)=N(v)\backslash(N_1(v)\cup N_2(v))$.
\end{itemize}

\noindent Similarly, for two distinct vertices $v,w\in V$, we can
also partition $N(v,w)=N(v)\cup N(w)$ into 3 disjoint sets as
follows,
\begin{itemize}
\item $N_1(v,w)=\{u\in N(v,w)\ |\ N(u)\backslash N[v,w]\neq\emptyset\}.$

\item $N_2(v,w)=\{u\in N(v,w)\backslash N_1(v,w)\ |\ N(u)\cap N_1(v,w)\neq\emptyset\}$.

\item $N_3(v,w)=N(v,w)\backslash(N_1(v,w)\cup N_2(v,w))$.
\end{itemize}

As stated in \cite{che:fer:kan:xia}, all the vertices of $G$ are
colored black initially. Reduction rules will color some vertices
white, which means that these white vertices are excluded from an
optimal dominating set of $G$. We repeat applying eight reduction
rules in \cite{che:fer:kan:xia} to reduce the planar graph until the
resulting graph, i.e. a \emph{reduced graph}, is unchanged.

\begin{Rule}\label{DS-reduce3}
If there is a black vertex $x\in N_2(v)\cup N_3(v)$ for each black
vertex $v$, then color $x$ white.
\end{Rule}

As Figure~\ref{fig1} exhibits, there is always a black vertex $u\in
V_1$ with $v\in N_2(u)\cup N_3(u)$ for any black vertex $v\in V_2$.
Thus for the initial graph $G_{\mathcal{H}}^K$, this rule must color
all the vertices of $V_2$ white. It implies that the optimal
dominating set $D$ of $G_{\mathcal{H}}^K$ contains no vertex in
$V_2$. In other words, $D\subseteq V_1$ is a cover of $\mathcal{H}$,
which means that those vertices colored white in $V_1$ can be
removed from $\mathcal{H}$.

Note that the rest rules either color vertices white or remove some
white vertices except the following two ones.

\begin{Rule}\label{DS-reduce1}
If $N_3(v)\neq\emptyset$ for some black vertex $v$, then
\begin{itemize}
\item Remove the vertices in $N_2(v)\cup N_3(v)$ from the current graph.

\item Add a new white vertex $v'$ and an edge $\{v,v'\}$.
\end{itemize}
\end{Rule}

\begin{Rule}\label{DS-reduce2}
If $N_3(v,w)\neq\emptyset$ for two black vertices $v$, $w$ and if
$N_3(v,w)$ cannot be dominated by a single vertex from $N_2(v,w)\cup
N_3(v,w)$, then
\begin{description}
\item[Case 1:]If $N_3(v,w)$ can be dominated by a single vertex from
$\{v,w\}$:
\begin{itemize}
\item If $N_3(v,w)\subseteq N(v)$ and $N_3(v,w)\subseteq N(w)$:
\begin{itemize}
\item Remove the vertices in $N_3(v,w)\cup(N_2(v,w)\cap N(v)\cap
N(w))$ from the current graph.

\item Add two new white vertices $z,z'$ and four edges
$\{v,z\}$, $\{w,z\}$, $\{v,z'\}$, $\{w,z'\}$.
\end{itemize}

\item If $N_3(v,w)\subseteq N(v)$ and $N_3(v,w)\nsubseteq N(w)$:
\begin{itemize}
\item Remove the vertices in $N_3(v,w)\cup(N_2(v,w)\cap N(v))$ from
the current graph.

\item Add a new white vertex $v'$ and an edge $\{v,v'\}$.
\end{itemize}

\item If $N_3(v,w)\nsubseteq N(v)$ and $N_3(v,w)\subseteq N(w)$:
\begin{itemize}
\item Remove the vertices in $N_3(v,w)\cup(N_2(v,w)\cap N(w))$ from
the current graph.

\item Add a new white vertex $w'$ and an edge $\{w,w'\}$.
\end{itemize}
\end{itemize}

\item[Case 2:]If $N_3(v,w)$ can not be dominated by a single vertex
from $\{v,w\}$:
\begin{itemize}
\item Remove the vertices in $N_3(v,w)\cup N_2(v,w)$ from the
current graph.

\item Add two new white vertices $v',w'$ and two edges $\{v,v'\}$,
$\{w,w'\}$.
\end{itemize}
\end{description}
\end{Rule}

It is obvious that both Rule~\ref{DS-reduce1} and
Rule~\ref{DS-reduce2} add new white vertices to the graph. We
categorize those new white vertices into $V_2$. In this sense, they
are viewed as self-loops in $\mathcal{H}$.

Therefore, we achieve a kernelization for the \textsc{Vertex Cover}
problem on $\mathcal{H}$.
\begin{enumerate}
\item Construct $G_{\mathcal{H}}$ and $G_{\mathcal{H}}^K$.

\item Reduce $G_{\mathcal{H}}^K$ to a reduced graph $G'$ by using
reduction rules in \cite{che:fer:kan:xia}.

\item Remove white vertices in $V_1$ from $G'$.

\item Remove all the edges between two vertices of $V_1$ from $G'$.

\item Construct a reduced hypergraph $\mathcal{H}'$ based on the
resulting bipartite graph.
\end{enumerate}

By Theorem~\ref{tchen}, the proof of Theorem~\ref{t1} is complete.

\subsection{Dual Problem}
In this section, we give a linear kernelization for the
\textsc{Independent Set} problem on a planar 3-uniform hypergraphs.

\begin{theorem}
Let $\mathcal{H}$ be a planar 3-uniform hypergraph. The
\textsc{Independent Set} problem admits a problem kernel of size
$40k$ on $\mathcal{H}$.
\end{theorem}

To prove the above theorem, we first need to introduce the
\textsc{Induced Matching} problem which is known to be $W[1]$-hard
\cite{kan:pel:xia:sch,mos:sik}.

A \emph{matching} $M$ in a graph $G=(V,E)$ is a subset of edges no
two of which have a common endpoint. If no two edges of $M$ are
joined by an edge of $G$, then $M$ is an \emph{induced matching}.

\begin{center}
\fbox{\parbox[c]{8cm}{\textsc{Induced Matching Problem}
\begin{description}
\item[Instance:]A graph $G=(V,E)$ and a nonnegative integer $k$.

\item[Parameter:]$k$.

\item[Problem:]Is there an induced matching $M\subseteq E$ with $|E|\leq
k$?
\end{description}
}}
\end{center}

\begin{theorem}[\cite{kan:pel:xia:sch}]
Let $G$ be a planar graph. The \textsc{Induced Matching} problem
admits a problem kernel of size $40k$ on $G$.
\end{theorem}

Since $\mathcal{H}$ is a planar 3-uniform hypergraph,
$G_{\mathcal{H}}^K$ must be planar. Hence, the rest task is to show
the following lemma.

\begin{lemma}
The \textsc{Independent Set} problem on a planar 3-uniform
hypergraph $\mathcal{H}$ is equivalent to the \textsc{Induced
Matching} problem on $G_{\mathcal{H}}^K$.
\end{lemma}

\begin{proof}
We assign a vertex $v_B$ in $G_{\mathcal{H}}^K$ for a hyperedge
$B\in\mathcal{H}$.
\begin{description}
\item[$(\Leftarrow)$] Let $M$ be an induced matching of
$G_{\mathcal{H}}$. Without loss of generality, let $I=\{u_{i1}\ |\
\{u_{i1},v_{B_i}\}\in M\}$. We claim that $I$ is an independent set
of $\mathcal{H}$. Otherwise, $\{u_{i1},u_{j1}\}\subseteq B$ for some
hyperedge $B$ with $i\neq j$. By definition of $G_{\mathcal{H}}^K$,
$\{u_{i1},v_{B_i}\}$ and $\{u_{j1},v_{B_j}\}$ are joined by
$\{u_{i1},u_{j1}\}$ in $G_{\mathcal{H}}^K$. It is a contradiction.

\item[$(\Rightarrow)$] Without loss of generality, let $I$ be an
independent set of $\mathcal{H}$. For $u_{i1},u_{j1}\in I$, we can
find two distinct hyperedges $B_i$ and $B_j$ such that $u_{i1}\in
B_i$ and $u_{j1}\in B_j$ by definition of independent set. It is
easy to see that $\{u_{i1},v_{B_i}\}$ and $\{u_{j1},v_{B_j}\}$ are
not joined by an edge in $G_{\mathcal{H}}^K$. Thus,
$M=\{u_{i1},v_{B_i}\}\ |\ u_i\in I\}$ is an induced matching of
$\mathcal{H}$.
\end{description}
We complete this proof.\qed
\end{proof}

By Theorem~\ref{dual}, we obtain lower bounds on kernel sizes for
the \textsc{Vertex Cover} problem and the \textsc{Independent Set}
problem on planar 3-uniform hypergraphs.

\begin{corollary}
For any $\varepsilon>0$, there is no
$(\frac{40}{39}-\varepsilon)k$-kernel for the \textsc{Vertex Cover}
problem on 3-uniform hypergraphs with bounded degree unless P$=$NP.
\end{corollary}

\begin{corollary}
For any $\varepsilon>0$, there is no
$(\frac{67}{66}-\varepsilon)k$-kernel for the \textsc{Independent
Set} problem on 3-uniform hypergraphs with bounded degree unless
P$=$NP.
\end{corollary}

\section{Conclusion}
In this paper we study kernelizations of the \textsc{Vertex Cover}
problem on 3-uniform hypergraphs, and show that this problem in
three classes of 3-uniform hypergraphs has linear kernels. We give
lower and upper bounds on the kernel size for them by the parametric
duality. An interesting open question is how to decide these
3-uniform hypergraphs efficiently. It is a challenge to explore more
3-uniform hypergraphs with less linear kernels in the future.
Additionally, designing better kernelization algorithms and pursuing
less kernel are also challenging tasks for the \textsc{Vertex Cover}
problem on 3-uniform hypergraphs.

\iffalse
\section*{Acknowledgement}
%
We would like to thank Prof. Yuxi Fu and Dr. Yijia Chen for their
valuable suggestions and discussion. This research is supported in
part by the National 973 Project (2003CB317005), the National Nature
Science Foundation of China (60573002, 60703033). \fi

\end{document}